\documentclass[a4paper,reqno]{amsart}%
\usepackage{amsthm,amsmath,amsfonts,amssymb,amsxtra,appendix,bookmark,dsfont,latexsym}
\usepackage{mathrsfs}
\usepackage{tikz}
\usetikzlibrary{arrows}
\usepackage{enumerate}
\usepackage{verbatim}

\usepackage{hyperref}
\usepackage{delarray,a4,color}
\usepackage{euscript}
\usepackage[latin1]{inputenc}

\usepackage[colorinlistoftodos]{todonotes}

\definecolor{darkred}{cmyk}{0,1,1,0.4}

\newcommand{\disp}{\displaystyle}

\theoremstyle{plain}
\newtheorem{theorem}{Theorem}[section]

\newtheorem{proposition}[theorem]{Proposition}

\theoremstyle{remark}

\numberwithin{equation}{section}

\newcommand{\aav}{\mathbf{A}}

\DeclareMathOperator{\supp}{supp}

\def\bq{\begin{eqnarray}}
\def\eq{\end{eqnarray}}
\def\bqq{\begin{eqnarray*}}
\def\eqq{\end{eqnarray*}}

\def\eps{\varepsilon}

\renewcommand{\epsilon}{\varepsilon}
\newcommand{\OO}{\mathcal{O}}
\newcommand{\daf}{\mathscr{D}^{\mathrm{af}}}
\newcommand{\eaf}{E^{\mathrm{af}}}
\newcommand{\uaf}{u^{\mathrm{af}}}
\newcommand{\maf}{\mu^{\mathrm{af}}_{\beta}}
\newcommand{\barho}{\bar{\rho}_{\beta}^{\mathrm{af}}}

\def\R {\mathbb{R}}

\def\cE {\mathcal{E}}

\def\B {\mathcal{B}}

\def\R {\mathbb{R}}

\renewcommand{\leq}{\leqslant}
\renewcommand{\geq}{\geqslant}

\newcommand{\bA}{\mathbf{A}}

\newcommand{\nablap}{\nabla^{\perp}}

\newcommand{\cEAF}{\cE ^{\mathrm{af}}}

\DeclareMathOperator{\curl}{\mathrm{curl}}

\newcommand{\cETF}{\mathcal{E} ^{\rm TF}}
\newcommand{\ETF}{E ^{\rm TF}}
\newcommand{\rhoTF}{\rho^{\rm TF}}
\newcommand{\lTF}{\lambda^{\rm TF}}

\numberwithin{equation}{section}
\newcommand{\bdm}{\begin{displaymath}}
\newcommand{\edm}{\end{displaymath}}
\newcommand{\bdn}{\begin{eqnarray}}
\newcommand{\edn}{\end{eqnarray}}
\newcommand{\bay}{\begin{array}{c}}
\newcommand{\eay}{\end{array}}
\newcommand{\ben}{\begin{enumerate}}
\newcommand{\een}{\end{enumerate}}
\newcommand{\beq}{\begin{equation}}
\newcommand{\eeq}{\end{equation}}
\newcommand{\beqn}{\begin{eqnarray}}
\newcommand{\eeqn}{\end{eqnarray}}
\newcommand{\bml}[1]{\begin{multline} #1 \end{multline}}
\newcommand{\bmln}[1]{\begin{multline*} #1 \end{multline*}}

\def\XXint#1#2#3{{\setbox0=\hbox{$#1{#2#3}{\int}$ }
\vcenter{\hbox{$#2#3$ }}\kern-.6\wd0}}

\newcommand{\one}{\mathds{1}}
\newcommand{\lf}{\left}
\newcommand{\ri}{\right}

\newcommand{\nv}{\mathbf{n}}

\newcommand{\rvp}{\mathbf{r}^{\prime}}
\newcommand{\rv}{\mathbf{r}}

\newcommand{\jv}{\mathbf{j}}

\newcommand{\diff}{\mathrm{d}}

\newcommand{\dtf}{\mathscr{D}^{\mathrm{TF}}}
\newcommand{\tfm}{\rho^{\mathrm{TF}}}

\newcommand{\tx}{\textstyle}
\renewcommand{\leq}{\leqslant}
\renewcommand{\geq}{\geqslant}

\title[Almost-Bosonic Anyons]{Local Density Approximation for Almost-Bosonic Anyons} 

\author{M. Correggi}
\address{Dipartimento di Matematica, ``Sapienza'' Universit\`{a} di Roma, P.le Aldo Moro, 5, 00185, Rome, Italy.}
\email{michele.correggi@gmail.com}
\urladdr{http://www1.mat.uniroma1.it/people/correggi/}

\author{D. Lundholm}
\address{KTH Royal Institute of Technology, Department of Mathematics, SE-100 44 Stockholm, Sweden}
\email{dogge@math.kth.se}

\author{N. Rougerie}
\address{CNRS \& Universit\'e Grenoble Alpes, LPMMC (UMR 5493), B.P. 166, F-38042 Grenoble, France}
\email{nicolas.rougerie@lpmmc.cnrs.fr}

\date{November 2016}

\begin{document}

\begin{abstract}
We discuss the average-field approximation for a trapped gas of non-interacting anyons in the quasi-bosonic regime. In the homogeneous case, i.e., for a confinement to a bounded region, we prove that the energy in the regime of large statistics parameter, i.e., for ``less-bosonic'' anyons, is independent of boundary conditions and of the shape of the domain. When a non-trivial trapping potential is present, we derive a local density approximation in terms of a Thomas-Fermi-like model. The results presented here mainly summarize \cite{CLR} with additional remarks and strengthening of some statements.
\end{abstract}

\maketitle


\section{Introduction}\label{sec:intro}

The physics of identical quantum particles in two-dimensions (2D) is much richer than in the three-dimensional world, where the symmetry under exchange admits only two opposite representations leading to {\it bosonic} (symmetric) and {\it fermionic} (antisymmetric) particles. However, in 2D, the way two or more particles are exchanged plays an important role in classifying the species of identical particles. There is indeed room for other statistics different from the usual Bose-Einstein and Fermi-Dirac ones, which go under the name of {\it intermediate} or {\it fractional} statistics. The corresponding particles are named {\it anyons} (a name originally introduced by F. Wilczek in \cite{Wilczek-82}) and include bosons and fermions as special cases. Mathematically speaking, the reason behind these unconventional features is the exchange symmetry: since the particles are indistinguishable the Hilbert space containing the states of the system must be the space of a one-dimensional\footnote{It is also possible to consider higher dimensional representations of the symmetry group, which lead to generalizations of anyons going under the name of {\it plektons} or {\it non-abelian anyons}.}  representation of the symmetry group. In three or more dimensions this group is the permutation group, which has only two inequivalent irreducible representations, given by the space of symmetric and antisymmetric functions. On the opposite, in 2D, there are many inequivalent ways of exchanging the particles  and the symmetry group to consider is the braid group which has infinitely many irreducible representations. Topologically, this difference is due to the fact that a circle is always contractible to a point in three or more dimensions, while it is not in 2D.

Another way of understanding the emergence of fractional statistics is simply by observing that for identical particles all the physics must be invariant under exchange. This has to be true at least for the probability density and therefore the modulus of the wave function $ |\Psi(\rv_1, \ldots, \rv_N)| $ can not change when two particles are exchanged:
\bdm
	  |\Psi(\ldots, \rv_i, \ldots, \rv_j, \ldots)| =  |\Psi(\ldots, \rv_j, \ldots, \rv_i, \ldots)|.
\edm
In 3D this implies that only the sign of $ \Psi $ can flip under exchange, while in two dimensions the wave function can acquire a generic phase factor, i.e.,
\beq
	\label{eq:anyonic gauge}
	\Psi(\ldots, \rv_i, \ldots, \rv_j, \ldots) = e^{i \pi \alpha} \Psi(\ldots, \rv_j, \ldots, \rv_i, \ldots),
\eeq
where $ \alpha \in \R $ can be {\it any} real number, which is called the {\it statistics parameter}. Without loss of generality one can however assume that
\beq
	\alpha \in (-1,1],
\eeq
where the points $ 0 $ and $ 1 $ (and $ -1 $ by periodicity) describe the usual bosonic and fermionic particles. 

Whether such anyonic particles do really exist is still a debated question within the physics community. Although elementary particles live in a three-dimensional world and therefore can never be anyonic, the possibility to observe some quasi-particles obeying to fractional statistics is more than plausible. In particular physical models involving anyonic particles have been proposed in relation to  the fractional quantum Hall effect. We refer to \cite{Khare-05,LunRou-16,Myrheim-99,Wilczek-90} and references therein for a more detailed discussion of this point.

In order to describe non-interacting anyonic particles in a trap, one has thus to face the problem of defining suitable Schr\"{o}dinger operators of the form $ \sum (- \Delta_i + V(\rv_i)) $ on a Hilbert space of functions satisfying \eqref{eq:anyonic gauge} \cite{DFT}. This poses hard technical questions since \eqref{eq:anyonic gauge} implies that anyonic wave functions are in general multi-valued. However, there is an alternative but equivalent approach which we are going to adopt in this note: instead of considering multi-valued functions satisfying \eqref{eq:anyonic gauge}, which goes under the name of {\it anyonic gauge}, one can rewrite the Hamiltonian in the {\it magnetic gauge}, i.e., set, at least formally,
\beq
	\label{eq:magnetic gauge}
	\Psi(\rv_1, \ldots, \rv_N) = :  \prod_{j < k} e^{i \alpha \phi_{jk}}	\Phi(\rv_1, \ldots, \rv_N),		
\eeq
where $ \phi_{jk} = \arg(\rv_j - \rv_k) $ is the angle between the vectors $ \rv_j $ and $ \rv_k $.
With such a choice it is easy to see that \eqref{eq:anyonic gauge} holds true if and only if $ \Phi \in L^2_{\mathrm{sym}}(\R^{2N}) $, i.e., $ \Phi $ must be {\it symmetric} under exchange, i.e., bosonic. There is however a price to pay for this simplification, i.e., the Schr\"{o}dinger operator acting on $ \Phi $ becomes
\beq
	\label{eq:ham}
	H  = \sum_{j = 1}^N \lf\{  \lf(- i \nabla_j + \alpha \aav_j \ri)^2 + V(\rv_j) \ri\},
\eeq
i.e., each particle carries a magnetic flux generated by the Aharonov-Bohm magnetic potential
\beq
	\label{eq:aav}	\aav_j(\rv_1, \ldots, \rv_N) = \sum_{k \neq j} \frac{\lf( \rv_j - \rv_k \ri)^{\perp}}{|\rv_j - \rv_k|^2},
\eeq
where $ \rv^{\perp} : = (-y,x)$.
To discuss self-adjointness of $ H $, a typical approach is to consider the symmetric operator defined by $ H $ on the domain of functions with support away from the diagonals $ \rv_j = \rv_k $. Such an operator admits several self-adjoint extensions (see, e.g., \cite{AT98, BS92} for the discussion of the two-particle case), but one can at least identify the Hamiltonian with its Friedrichs extension \cite[Sect. 2.2]{LS13}.

An alternative approach to the rigorous definition of anyonic Schr\"{o}dinger operators is through a sort of regularization of \eqref{eq:ham}: since the singularities of $ H $ live on the diagonals $ \rv_j = \rv_k $, it seems reasonable to replace the singular magnetic potential \eqref{eq:aav} with its cut-off version obtained as
\beq
	\label{eq:avv extended}
	\aav^{R}_j(\rv_1, \ldots, \rv_N) = \sum_{k \neq j} \lf( \rv_j - \rv_k \ri)^{\perp}	\lf[ \frac{\one_{\{|\rv_j - \rv_k| > R\}}}{|\rv_j - \rv_k|^2} + \frac{\one_{\{|\rv_j - \rv_k| \leq R\}}}{R^2} \ri],	
\eeq
for some $ R > 0 $. In other words the point-like anyons are replaced with extended particles exactly in the same way in classical electrodynamics giving a spatial dimension to point charges makes the theory more regular. It is thus clear why the so-obtained particles are named {\it extended anyons}. The corresponding Schr\"{o}dinger operator $ H_{R} $ is essentially self-adjoint on smooth functions with compact support. We refer to \cite{LarLun-16, LunRou-15, Lundholm-16} and references therein for a detailed discussion of extended anyon models.

\subsection{The model: average-field approximation}
The physical system we plan to describe in this note is a many-body gas of a large number $N $ of non-interacting anyons confined in a trap. When $ N $ gets very large, the quantum mechanical description above becomes very inefficient and the complexity of the corresponding equations calls for a simplified effective model. Such a problem has already been faced in the physics literature \cite{IenLec-92,Wilczek-90}, where it was noted that, if the statistic parameter $ \alpha $ is close to $ 0 $ (bosons\footnote{A similar approximation is often applied in the opposite regime of $ \alpha $ close to $ 1 $, i.e., for almost fermionic anyons. A one-particle effective model is expected to emerge in this case too, although the expression of the function may be different.}), the mean-field behavior of the anyonic gas is expected to be suitably approximated by the {\it average-field energy functional}\footnote{In fact the explicit expression of the average-field functional has never appeared in the physics literature, where a further approximation is made (see Section \ref{sec:main}). In this respect $ \cEAF_{\beta} $ was first introduced in \cite{LunRou-15}.}, which reads
\begin{equation}
	\label{eq:avg func}
	\framebox{
	$ \cEAF_\beta[u] := \disp\int_{\R ^2} \diff \rv \: \left\{ \left| \left( -i\nabla +  \beta \bA [|u|^2] \right) u \right|^2 + V(\rv) |u|^2 \right \} $} 
\end{equation}
where $V $  is the trapping potential and 
\begin{equation}\label{eq:avg field}
	\bA [\rho] := \nablap w_0 \ast \rho, 
	\qquad w_0 (\rv) := \log |\rv|,
\end{equation}
and $ \nabla^\perp := (-\partial_y,\partial_x) $. Since $ w_0 $ is the fundamental solution of the Laplace equation $ \Delta w_0 = 2 \pi \delta(\rv) $, the magnetic field associated to $ \bA [\rho] $ is
\beq
	\lf( \curl \bA[\rho] \ri) (\rv) = 2\pi \rho(\rv),
\eeq
i.e., its intensity is proportional to the particle density 
$\rho(\rv) = |u(\rv)|^2$. 
The parameter $ \beta \in \R $ is related to the statistics parameter $ \alpha $ in a way that will be made precise below. The energy \eqref{eq:avg func} can indeed be easily obtained by evaluating the expectation value of the many-body Schr\"{o}dinger operator \eqref{eq:ham} on a product state $ \Psi(\rv_1, \ldots, \rv_N) = u(\rv_1) \cdot \cdots \cdot u(\rv_N) $ and dropping any term vanishing in the limit $ N  \to \infty $. The parameter $ \beta $ is obtained as the limit of $ N \alpha $ and therefore one is forced to assume that $ \alpha = \OO(N^{-1}) $ to obtain a finite energy in the limit, i.e., $ \alpha \to 0  $ and we are considering almost-bosonic anyons. A similar behavior can be recovered in the opposite regime of almost-fermionic anyons, i.e., when $ \alpha \to  1 $ \cite{Wilczek-90}. 

This heuristic derivation of the average-field functional was in fact made rigorous in \cite[Theorem 1.1]{LunRou-15}, where it is proven that, assuming 
\beq
	\label{eq:alpha}
	\alpha = \frac{\beta}{N-1},
\eeq
in the limit $ N \to \infty $, the ground state energy per particle of a gas of non-interacting anyons converges to the infimum of the average-field functional \eqref{eq:avg func}. The result is actually proven for extended anyons but with a particle radius $ \sim N^{-\gamma} $, for some $ \gamma > 0 $, i.e., the limit $ N \to \infty $ describes also the convergence to point-like particles. In addition the $ k$-th particle reduced density matrix associated to any quasi-ground state converges to a convex combination of projectors onto the tensor product of minimizers of $ \cEAF_{\beta} $. Hence one can say that the minimization of the average-field functional (see below) provides a good approximation for the ground state behavior of the many-body anyon gas in this particular limit. 

The effective theory provided by the average-field energy depends on only one physical parameter $ \beta \in \R $, which in fact can be assumed to be positive
\beq
	\beta \geq 0,
\eeq
thanks to the symmetry under exchange $ u \to u^* $ (complex conjugation) of the energy. In spite of a certain analogy with other effective nonlinear theories applied to the description of, e.g., Bose-Einstein condensates (Gross-Pitaevskii theory) or superconductors (Ginzburg-Landau theory), the specificity of the average-field functional lies in the special form of the nonlinearity, which appears only in the magnetic potential $ \aav[|u|^2] $ and therefore affects mostly the phase of the effective wave function $ u $. In addition, the nonlinearity in $ \cEAF_{\beta}[u] $ is actually nonlocal and, as it becomes apparent in its Euler-Lagrange equation \eqref{eq:var eq}, it generates cubic quasi-linear and quintic semi-linear terms. 

Given the physical meaning of the parameter $ \beta $ and \eqref{eq:alpha}, we are thus describing a model of anyonic behavior which is a ``small perturbation" of the usual bosonic one. In the regime $ \beta \to 0 $ one thus expects to recover conventional bosons, and this can actually be proven rigorously, at least for the ground state properties \cite[Proposition 3.8]{LunRou-15}. More interesting and unexplored is the opposite regime
\beq
	\label{eq:beta}
	\beta \gg  1,
\eeq
where the anyonic features of the model should emerge. This is indeed the asymptotic regime we are going to discuss in this note in the framework of the average-field approximation.

\subsection{Minimization of the average-field functional} Let us now discuss closer the average-field functional from the mathematical view point.  First of all we assume that the trapping potential is positive
\beq
	\label{eq:V 1}
	V(\rv) \geq 0,
\eeq
which does not imply any loss of generality but simply a change of energy scale. In the following we will restrict our attention to a smaller class of smooth homogeneous potentials such that
\beq
	\label{eq:V 2}
	V(\lambda \rv) = \lambda^s V(\rv),		\qquad		s > 1,
\eeq
for any $ \lambda > 0 $. The potential must also be trapping and therefore we assume that
\beq
	\label{eq:V 3}
	\lim_{R \to \infty} \min_{|\rv| \geq R} V(\rv) = + \infty.
\eeq
A typical case is the (anisotropic) harmonic potential $ V(\rv) = a x^2 + b y^2 $, $ a,b \in \R^+ $, but we are also going to consider the flat case, i.e., formally $ s = + \infty $ or
\beq
	\label{eq:V homo}
	V(\rv) = 
	\begin{cases}
		0,		&	\mbox{ in } \Omega \subset \R^2,	\\
		+ \infty	&	\mbox{ otherwise},
	\end{cases}
\eeq
for some simply connected domain $ \Omega $ with Lipschitz boundary. The average-field functional in this case simply becomes
\beq
	\label{eq:avgf homo}
	\cEAF_{\beta, \Omega}[u] := \disp\int_{\Omega} \diff \rv \: \left| \left( -i\nabla +  \beta \bA [|u|^2] \right) u \right|^2. 
\eeq
Note that the choice \eqref{eq:V homo} would naturally lead to consider Dirichlet conditions on the boundary of $ \Omega $, while the Neumann case seems to be more appropriate to describe the homogeneous setting. As a matter of fact we are going to prove that boundary conditions do not matter at all in the limit $ \beta  \to \infty $ and we keep referring to both cases as the {\it homogenous} anyonic gas. 

The energy functional \eqref{eq:avg field} is well defined for any $ u \in H^1(\R^2) $ such that $ V |u|^2 \in L^1(\R^2) $, as it can be easily seen thanks to the inequality \cite[Lemma 3.4]{LunRou-15}
\bdm
	\int_{\R^2} \diff \rv \: \lf| \aav\lf[|u|^2\ri] \ri|^2 \lf|u\ri|^2 \leq \tx\frac{3}{2} \lf\| u \ri\|_2^4 \lf\| \nabla |u| \ri\|_2^2,
\edm
which allows to bound the most singular term in the energy by the kinetic term. We thus consider the minimization domain
\beq
	\label{eq:daf}
	 \daf : = \lf\{ u \in H^1(\R^2) \: \big| V |u|^2 \in L^1(\R^2), \lf\| u \ri\|_2 = 1 \ri\},
\eeq
where the $ L^2$ normalization has the usual meaning of the quantum mechanical probability conservation, and set
\beq
	\label{eq:eaf}
	\eaf_{\beta} : = \inf_{u \in \daf} \cEAF_\beta[u].
\eeq
The boundedness from below of $ \eaf_{\beta} $, in fact its positivity, is a simple consequence of the assumption on $ V $. The existence of a corresponding minimizing $ \uaf_{\beta} $ does not immediately follow but can be proven exploiting the estimate above \cite[Proposition 3.7]{LunRou-15} together with the following magnetic bounds \cite[Lemma 3.4]{LunRou-15}: for any $ \beta \in \R $ and $ u \in H^1(\R^2) $,
\bdn
	\int_{\R^2} \diff \rv \:\left| \left( -i\nabla +  \beta \bA [|u|^2] \right) u \right|^2 &\geq &  \int_{\R^2} \diff \rv \: \lf| \nabla |u| \ri|^2,	\label{eq:diamag ineq}	\\
	\int_{\R^2} \diff \rv \:\left| \left( -i\nabla +  \beta \bA [|u|^2] \right) u \right|^2 & \geq  & 2\pi |\beta| \int_{\R^2} \diff \rv \: \lf| u\ri|^4. \label{eq:magn ineq}
\edn
Both inequalities can actually be generalized \cite[Lemma 3.2]{CLR} to any domain $ \Omega \subset \R^2 $, but the second one requires the additional assumption $ u \in H^1_0(\Omega) $, i.e., the support of $ u $ must be strictly contained in $ \Omega $. Note also that the first bound is simply the extension to the nonlinear vector potential $ \aav[|u|^2] $ of the usual diamagnetic inequality \cite[Theorem 7.21]{LieLos-01}.

For the homogeneous gas we have to specify the dependence on boundary conditions and we thus set
\bdn
	\daf_{\mathrm{N}} &: =& \lf\{ u \in H^1(\Omega) \: \big| \lf\| u \ri\|_2 = 1 \ri\},	\\
	\daf_{\mathrm{D}} &: =& \lf\{ u \in H_0^1(\Omega) \: \big| \lf\| u \ri\|_2 = 1 \ri\},
\edn
and correspondingly
\beq
	\eaf_{\beta, \mathrm{N}/\mathrm{D}} : = \inf_{u \in \daf_{\mathrm{N}/\mathrm{D}}} \cEAF_{\beta, \Omega}[u].
\eeq

For any trapping potential $ V $ satisfying the above assumptions, one can derive the variational equation associated to the minimization of $ \cEAF_{\beta} $ \cite[Appendix A]{CLR}, which reads
\begin{equation} \label{eq:var eq}
	\left[ 
		\left( -i\nabla + \beta\bA[|u|^2] \right)^2 + V
		-2\beta \nablap w_0 * \left( \beta\bA[|u|^2]|u|^2 + \jv[u] \right)
		\right]u = \maf u,
\end{equation}
where the current $ \jv[u] $ is given by
\beq
	\label{eq:jv}
	\jv[u] : = \tx\frac{i}{2} \lf( u \nabla u^* - u^* \nabla u \ri).
\eeq
The chemical potential $ \maf $ can be expressed in terms of the ground state energy $ \eaf_{\beta} $ and the corresponding minimizer as
\bml{
	\maf = \eaf_{\beta} + \int_{\R^2} \diff \rv \:
		\left\{ 2\beta \bA\lf[ |\uaf_{\beta} |^2\ri] \cdot \jv \lf[\uaf_{\beta}\ri] + 2\beta^2 \lf|\bA\lf[|\uaf_{\beta}|^2\ri] \ri|^2 |\uaf_{\beta}|^2 \right\}
	\label{eq:maf}\\
	= \int_{\R^2} \diff \rv \: \left\{ 
		\lf|\nabla \uaf_{\beta}\ri|^2 + V|\uaf_{\beta}|^2
		+ 4\beta \bA\lf[|\uaf_{\beta}|^2\ri] \cdot \jv\lf[\uaf_{\beta}\ri]
		+ 3 \beta^2 \lf| \bA\lf[|\uaf_{\beta}|^2\ri] \ri|^2 |\uaf_{\beta}|^2
		\right\}. 
}

\section{Main Results}
\label{sec:main}

As anticipated we are going to study the ground state properties of a trapped gas of non-interacting anyons in the average-field approximation, i.e., we will investigate the minimization of the effective energy functional \eqref{eq:avg func} (or its homogenous counterpart \eqref{eq:avgf homo}) in the asymptotic regime $ \beta \gg 1 $.

In physics literature the average-field approximation is usually performed by simplifying further the problem and assuming that a local anyonic density $ \rho(\rv) $ would generate a magnetic field 
\bdm
	\mathbf{B}(\rv) \simeq  2 \pi N |\alpha| \rho(\rv) \simeq 2 \pi \beta \rho(\rv) 
\edm
proportional to the density itself. This is clearly inspired by the self-generated magnetic potential appearing in the functional $ \cEAF $ and is very often applied to the homogenous case, i.e., when $ \rho = \mathrm{const.} $ in some bounded region. In the lowest energy state (lowest Landau level) the energy per particle should then be given by the magnetic energy $ |\mathbf{B}| $, to be averaged over the local density $ \rho(\rv) $ itself. Taking into account the potential energy, one thus recovers the Thomas-Fermi (TF)-like functional
\beq
	\label{eq:TF approx}
	\int_{\R^2} \diff \rv \: \lf\{ 2 \pi \beta \rho^2(\rv) + V \rho(\rv) \ri\}.
\eeq
It is interesting to remark that the form of the energy above is very close to the one obtained via the conventional Thomas-Fermi approximation for two-dimensional fermions: recalling that $ \rho $ is normalized to 1, the TF energy per particle of $ N $ fermions is
\bdm
	\int_{\R^2} \diff \rv \: \lf\{ 2 \pi N \rho^2(\rv) + V \rho(\rv) \ri\}.
\edm
Since fermions are identified by the choice $ \alpha = 1 $, it is not surprising that the energy per particle of non-interacting anyons would be given by
\bdm
	\int_{\R^2} \diff \rv \: \lf\{ 2 \pi \alpha N \rho^2(\rv) + V \rho(\rv) \ri\},
\edm
which is close to \eqref{eq:TF approx}, when $ \alpha \sim \beta/N $.

\subsection{Homogeneous anyonic gas}\label{sec:homo} When the gas is confined to some bounded region $ \Omega $, where $ V $ is constant (e.g., zero), the optimal density minimizing the energy \eqref{eq:TF approx} is simply the constant function $ \rho(\rv) = |\Omega|^{-1} $ and the corresponding energy is thus $ \mathrm{const.} \: \beta $. According to \eqref{eq:TF approx} the precise prefactor should be $ 2 \pi |\Omega|^{-1} $, which simply amounts to $ 2\pi $ when $ \Omega $ has unit area $ |\Omega| = 1 $.

In our main result about the homogenous gas we are going to recover the linear dependence of $  \eaf_{\beta, \mathrm{N}/\mathrm{D}} $ on $ \beta $ and prove that the expression is independent of boundary conditions and of the shape of the domain. However, whether the prefactor is $ 2\pi $ or larger remains an open question (see the discussion below).

	\begin{theorem}[Energy asymptotics]
		\label{teo:eaf homo}
		Let $\Omega \subset \R ^2$ be a bounded simply connected domain with Lipschitz boundary, then
		\beq
			\label{eq:eaf homo}
			\frac{\eaf_{\beta, \mathrm{N}}}{\beta} = \frac{\eaf_{\beta,\mathrm{D}}}{\beta} + \OO\big(\beta^{-1/7+\eps}\big),	
		\eeq
		and the limits of both quantities coincide and are finite. Moreover we have 
		\beq
			\label{eq:coefficient}
			|\Omega| \lim_{\beta \to + \infty} \frac{\eaf_{\beta, \mathrm{N}/\mathrm{D}}}{\beta} = : e(1,1) \geq 2\pi.
		\eeq
	\end{theorem}
	
The reason why we define the coefficient $ e(1,1) $ as a function depending on two parameters is that one can think of a more general limit in which the rescaled statistic parameter $ \beta $ is set equal to $ \gamma/\eps $, $ \gamma > 0  $ and $ 0 \leq \eps \ll 1 $, and define (the dependence on $ |\Omega|^{-1/2} $ is chosen for further convenience)
\beq
	e\big(\gamma, |\Omega|^{-1/2}\big) := \lim_{\eps \to 0^+} \eps \eaf_{\frac{\gamma}{\eps}, \mathrm{N}/ \mathrm{D}},
\eeq
which reduces to \eqref{eq:eaf homo} when $ \Omega $ has unit area and $ \gamma = 1 $. Equivalently, one can define $ e(\beta,\rho) $ as the energy per unit area in a thermodynamic limit $ L \to \infty $ on a rescaled domain $ L \Omega $, under normalization $ \lf\| u \ri\|_2^2 = \lambda L^2 |\Omega| $, $ \lambda > 0 $ (see also Section \ref{sec:homo proofs}). The two quantities, which are a priori different, equal because of the following scaling property of the functional: let $ \lambda, \mu \in \R^+ $ and set $ u_{\lambda,\mu}(\rv) := \lambda u(\rv/\mu) $, then
\beq
	\label{eq:scaling}
	\cEAF_{\beta,\mu\Omega}[u_{\lambda,\mu}] = \lambda^2 \cEAF_{ \lambda^2 \mu^2\beta,\Omega}[u],
\eeq
where $ \mu \Omega $ stands for the dilated domain $ \{ \mu \rv \: | \: \rv \in \Omega \} $. Once applied to $ e\lf(\gamma, |\Omega|^{-1}\ri) $, this scaling law implies that
\beq
	e\big(\gamma, |\Omega|^{-1/2}\big) = \gamma |\Omega|^{-1} e(1,1),
\eeq
and therefore $ e(1,1) $ determines the function entirely.
	
The bound \eqref{eq:coefficient} is not expected to be optimal, i.e., we conjecture that the strict inequality holds true
\beq
	e(1,1) > 2\pi.
\eeq
The estimate \eqref{eq:coefficient} is indeed a consequence of \eqref{eq:magn ineq}, which yields for any $ u \in H^1_0(\Omega) $, i.e., satisfying Dirichlet boundary conditions, 
\beq
	\cEAF_{\beta,\Omega}[u] \geq 2 \pi \beta \lf\| u \ri\|_4^4.
\eeq
The r.h.s. is saturated by the constant function $ u(\rv) = |\Omega|^{-1/2} $ (although it does not satisfy Dirichlet boundary conditions), yielding the estimate \eqref{eq:coefficient}, via \eqref{eq:eaf homo}. Even neglecting the problem of boundary conditions, there are strong indications that the constant function is very far from a minimizer of the functional for large $ \beta $: in order to minimize the  contribution of the magnetic field and its huge circulation, it is indeed much more convenient to distribute more or less uniformly a large number of vortices, whose fluxes compensate $ \aav[|u|^2] $. The picture would resemble then what is expected for fast rotating Bose-Einstein condensates \cite{CY08,CPRY12,CR13} or for superconductors in strong magnetic fields \cite{SS07}, i.e., the occurrence of vortices on a regular lattice (Abrikosov lattice). In this case, $ u $ must vanish at the center of each vortex and, even though $ u $ can still be close to a constant at larger scales in a weak sense (e.g., in $ L^p $, $ p < \infty $), the vortex core being very small, the interaction energy between vortices is expected to make the inequality \eqref{eq:coefficient} strict. It is interesting to note that such a behavior for anyons has already been conjectured in \cite[p.~1012]{ChenWilWitHal-89}, although the consequence on $ e(1,1) $ has not been noticed. We plan to investigate further this question in a future work.

The homogeneity of the system is confirmed by the following result about the density, which is independent of boundary conditions. In view of the discussion above, it is worth remarking that our next estimate  \eqref{eq:uaf homo} is perfectly compatible with the presence of a huge number of vortices with small core.

	\begin{theorem}[Density asymptotics]
		\label{teo:uaf homo}
		Under the same hypothesis of Theorem \ref{teo:eaf homo} and for any minimizer $ \uaf_{\beta, \mathrm{N}/\mathrm{D}} $ of $ \cEAF_{\beta,\Omega} $ in $ \daf_{\mathrm{N}/\mathrm{D}} $,
		\beq
			\label{eq:uaf homo}
	 		|\Omega|^{1/2} \lf|\uaf_{\beta, \mathrm{N}/\mathrm{D}}\ri| \xrightarrow[\beta \to \infty]{\lf(C^{0,1}_0(\Omega) \ri)^*} 1,
		\eeq
		where $C^{0,1}_0(\Omega)$ is the space of Lipschitz functions vanishing on $ \partial\Omega $.
	\end{theorem}
	
The above result guarantees that on scale one, i.e., on the scale of the domain $ |\Omega| $, any minimizer of the average-field functional \eqref{eq:avgf homo} can be well approximated by a constant. In fact, we expect this to be true even at finer scales, which are much larger than the vortex spacing. This latter characteristic length is of order $ \beta^{-1/2} $ for the homogeneous gas, since the circulation to compensate is of order $ \beta $ and the optimal vortex distribution is thought to be given by a regular lattice of singly quantized vortices, thus leading to an average spacing of order $ \beta^{-1/2} $. Hence, on any length scale much larger than $ \beta^{-1/2} $ the approximation by a constant should be accurate. Next result, which is certainly not optimal (one would expect it to hold true up to $ \eta < 1/2 $), shows however an instance of this behavior. We denote by $ \Omega^{\circ} $ the interior of the domain $ \Omega $.
	
	\begin{proposition}[Local density approximation]
		\label{pro:uaf homo 2}
		Let $ \rv_0 \in \Omega^{\circ} $, and $ R > 0 $ finite, then\footnote{We denote by $ \B_R : = \B_R(0) $ a ball of radius $ R > 0 $ centered at the origin. An analogous ball centered at $ \rv $ will be denoted by $ \B_R(\rv) $.} 
under the same hypothesis of Theorem \ref{teo:eaf homo},
		\beq
			\label{eq:uaf homo 2}
			 \lf|\uaf_{\beta, \mathrm{N}/\mathrm{D}}\lf(\rv_0 + \beta^{-\eta} \: \cdot \: \ri) \ri|   \xrightarrow[\beta \to + \infty]{\lf(C^{0,1}_0(\B_R) \ri)^*}  |\Omega|^{-1/2},	\qquad		\mbox{for any } 0 < \eta < \tx\frac{1}{14}.
		\eeq
	\end{proposition}
	
\subsection{Trapped anyonic gas} We now consider a gas of non-interacting anyons which is trapped by a more general potential $ V(\rv) $ satisfying the assumptions \eqref{eq:V 1}--\eqref{eq:V 3}. As anticipated in the Introduction, the paradigm is given by the anisotropic harmonic oscillator $ V(\rv) = a x^2 + b y^2 $, which contains the usual symmetric oscillator as a special case. Indeed, we do not require radial symmetry of $ V $ but only homogeneity of degree $ s > 1 $ as a function of $ \rv $. This last assumption can be relaxed as well and our result applies to the more general class of asymptotically homogeneous potentials as defined in \cite[Definition 1.1]{LieSeiYng-01}. We stick however to the hypothesis above for the sake of concreteness. The relevance of the scaling property $ V(\lambda \rv) = \lambda^s V(\rv) $ will become apparent when we will discuss the effective TF theory describing the behavior of the average-field functional in the regime $ \beta \gg 1 $: the so-obtained TF functional indeed admits a scaling property which allows to factor out the dependence on $ \beta $, provided $ V $ is a homogeneous function (or an asymptotically homogeneous function for $ \beta $ large).

The heuristics for the asymptotics $ \beta \to \infty $ of the functional $ \cEAF_{\beta} $ relies on the results proven in the previous Section: it is reasonable to assume that locally, on a suitable fine scale, the inhomogeneity generated by $ V(\rv) $ does not play any role and therefore the energy on that scale is given by the one of the homogeneous gas, i.e., $ \simeq e(\beta,\rho) = e(1,1) \beta \rho^2 $. Hence, the TF effective functional we expect to recover in the regime $ \beta \gg 1 $ must have the form
\begin{equation}
	\label{eq:tff}
	\cETF_\beta [\rho] := \int_{\R ^2} \diff \rv \lf\{  e(1,1) \beta \rho^2 + V \rho \right\}, 
\end{equation}
where $ \rho : = |u|^2 $ is the gas density. The ground state energy of the TF functional is defined as
\begin{equation}\label{eq:tfe}
	\ETF_\beta := \min_{\rho \in \dtf}  \cETF_\beta [\rho]
\end{equation}
with
\beq
	\label{eq:dtf}
	\dtf : = \lf\{ \rho \in L^2(\R^2) \: \big| \: \rho \geq 0, \lf\| \rho \ri\|_1 = 1 \ri\}.
\eeq
The associated (unique) minimizer is denoted by $ \tfm_{\beta}$. Under the hypothesis \eqref{eq:V 1}--\eqref{eq:V 3}, the minimization \eqref{eq:tfe} is actually explicit and one can extract the $ \beta $ dependence simply by rescaling the lengths. The result is
\beq
	\label{eq:TF scaling}
	\ETF_\beta = \beta ^{s/(s+2)} \ETF_1, \qquad \tfm_\beta(\rv) = \beta^{-2/(2+s)} \rhoTF_1 \left( \beta ^{-1/(s+2)} \rv \right),
\end{equation}
where $ \tfm_1 $ minimizes the TF energy with $ \beta = 1 $ and is given by the compactly supported function
\beq
	\label{eq:tfm1}
	\tfm_1(\rv) = \frac{1}{2 e(1,1)} \left[ \lTF_1 - V(\rv)\right]_+,
\end{equation}
with  chemical potential $ \lTF_1 = \ETF_1 + e(1,1) \lf\| \tfm_1 \right\|_2^2 > 0 $. If the potential $ V $ was radial, then the support of the rescaled TF minimizer $ \tfm_1 $ would be a disc of radius $ \propto (\lTF_1)^{1/s} $. In the general case the support of $ \tfm_1 $ is more complicated but still contained in a ball of finite radius $ R_0 $, 
so that
\beq
	\supp \lf( \tfm_{\beta} \ri) \subset \B_{R_0 \beta^{1/(2+s)}}.
\eeq
Without loss of generality we also assume that $ \partial_{\nv} V \neq 0 $ a.e. along $ \partial \supp \tfm_1 $, with $ \nv $ the outward normal to the boundary.

The following result shows that the TF functional \eqref{eq:tff} provides indeed a fine approximation of the average-field energy:

	\begin{theorem}[Energy asymptotics]
		\label{teo:eaf asympt}
		Let $ V $ satisfy the assumptions \eqref{eq:V 1}--\eqref{eq:V 3}, then as $ \beta \to + \infty $
		\beq
			\label{eq:eaf asympt}
			\beta ^{-s/(s+2)} \eaf_{\beta} = \ETF_1 + o(1).
		\eeq
	\end{theorem}
	
The analogue of Theorem \ref{teo:uaf homo} describing the asymptotics of the average-field density as $ \beta \to + \infty $ is:

	\begin{theorem}[Local density approximation]
		\label{teo:uaf asympt}
		Under the same hypothesis of Theorem \ref{teo:eaf asympt} and for any minimizer $ \uaf_{\beta} $ of $ \cEAF_{\beta} $ in $ \daf $, we have for any $ \rv \in \R^2 $ and $ R > 0 $
		\beq
			\label{eq:uaf asympt}
	 		 \beta^{2/(s+2)} \lf|\uaf_{\beta}\lf( \beta^{1/(s+2)} \, \cdot \, \ri) \ri|^2  \xrightarrow[\beta \to \infty]{\lf(C^{0,1}_0(\B_R) \ri)^*} \tfm_1( \, \cdot \, ).
		\eeq		
	\end{theorem}
	
	As for the homogeneous gas, the above estimate is expected to holds true on any finer scale, which is much larger than the vortex one, i.e., the rescaled average-field density should be close to $ \tfm_1 $ even locally in any small ball of radius $ \beta^{-\eta} $, $ \eta < s/(2(s+2)) $, which is the mean spacing of the conjectured vortex lattice. In fact, by exploiting the explicit remainder of the energy asymptotics \eqref{eq:eaf asympt} derived in the proofs, we could prove that \eqref{eq:uaf asympt} holds on a scale shorter than the one of $ \tfm_{\beta} $ but still much larger than the optimal one, exactly as in Proposition \ref{pro:uaf homo 2}. We skip the statement for the sake of brevity.

\section{Sketch of the Proofs}

We present here a synthetic exposition of the main arguments used in the proofs of the results stated in the previous Section. We refer to \cite{CLR} for further details. The starting point is the discussion of the homogeneous gas (see Section \ref{sec:homo}), which will be used as a key tool to take into account the inhomogeneity introduced by the trapping potential $ V $. 

\subsection{Homogeneous gas} \label{sec:homo proofs} Exploiting the scaling property \eqref{eq:scaling}, it is possible to show that the large $ \beta $ limit in a fixed domain $ \Omega $ is equivalent to a thermodynamic limit $ L \to \infty $ of a domain $ L\Omega $ with normalization $ \lf\| u \ri\|_2^2 \propto L^2 |\Omega| $. Explicitly, for any $ \gamma \geq 0 $,
\beq
	\label{eq:td limit identity}
	\lim_{\beta \to + \infty} \frac{\eaf_{\gamma\beta,\mathrm{N}}}{\beta} = \lim_{L \to + \infty} \frac{1}{\lambda^2 L^2|\Omega|^2} \inf_{u \in H^1(L\Omega), \lf\| u \ri\|_2^2 = \lambda L^2 |\Omega|} \cEAF_{\gamma,L\Omega}[u],
\eeq
where $ \lambda > 0 $ is a positive parameter, which is kept fixed as $ L \to + \infty $, i.e., we are considering a large volume limit with fixed density
\bdm
	\lambda = \frac{\lf\| u \ri\|_{L^2(L\Omega)}^2}{|\Omega| L^2}.
\edm  
The relation between $ \beta $ and $ L $ in the identity above is $ \beta = : \lambda |\Omega| L^2 $ and $ \gamma $ plays the role of a rescaled statistic parameter. In the Dirichlet case the definition is perfectly analogous and the only difference is that $ H^1(L\Omega) $ must be replaced with $ H^1_0(L\Omega) $ on the r.h.s..

In \cite{CLR} Theorem \ref{teo:eaf homo} (and consequently Theorem \ref{teo:uaf homo}) is proven by a direct inspection of the large $L$ limit of the r.h.s. of \eqref{eq:td limit identity}. The main steps are the following:
\begin{enumerate}[(i)]
	\item a priori bound on the r.h.s. of \eqref{eq:td limit identity}, showing that it is a bounded quantity, which allows to define $ e(\gamma,\lambda) $ at least as
		\beq
			\label{eq:e}
			e(\gamma,\lambda) : = \liminf_{L \to + \infty} \frac{1}{L^2|\Omega|} \inf_{u \in H^1(L\Omega), \lf\| u \ri\|_2^2 = \lambda L^2 |\Omega|} \cEAF_{\gamma,L\Omega}[u],
		\eeq
		since we do not know at this stage whether the limit does exist;
	\item use of the scaling property \eqref{eq:scaling}, which yields
		\beq
			e(\gamma,\lambda) = \gamma \lambda^2 e(1,1);
		\eeq
	\item proof of the existence of the thermodynamic limit $ L \to \infty $ when $ \Omega $ is a (unit) square;
	\item comparison of the Dirichlet and Neumann energies for squares showing that they coincide in the limit;
	\item extension of the existence of the thermodynamic limit to any domain $ \Omega $. For this last step it suffices to prove that the limit of the Dirichlet energy equals $ e(\gamma,\lambda)$, because the result then follows from the definition \eqref{eq:e} and the usual bound $ \eaf_{\beta, \mathrm{N}} \leq \eaf_{\beta, \mathrm{D}} $.
\een

We are not going to discuss all the details of the steps above but the first one (see Proposition \ref{pro:upper bound} below). We just point out that the most non-trivial result is the comparison Dirichlet/Neumann, whose fundamental tool is an IMS-type localization formula \cite[Lemma 3.3]{CLR} for a suitable partition of unity, together with the estimates  \eqref{eq:diamag ineq} and \eqref{eq:magn ineq}. The initial restriction to squares is motivated by the simplicity in constructing suitable partitions and coverings.

The a priori bound showing that the energy $ \eaf_{\beta,\Omega} $ grows at most linearly in $ \beta $ is proven in \cite[Lemma 3.1]{CLR} in the same thermodynamic setting discussed above. We state it here in a different but equivalent form:

	\begin{proposition}[Trial upper bound]
		\label{pro:upper bound}
		Under the same hypothesis as in Theorem \ref{teo:eaf homo}, there exists a finite constant $ C > 0  $ such that as $  \beta \to \infty $
		\beq
			\label{eq:ub homo}
			\frac{\eaf_{\beta, \mathrm{N}}}{\beta} \leq \frac{\eaf_{\beta, \mathrm{D}}}{\beta} \leq C.
		\eeq
	\end{proposition}
	
	\begin{proof}
		The inequality between the Neumann and Dirichlet energies is trivial, since the minimization domain in the latter case is smaller. Assume then without loss of generality (thanks to the scaling property \eqref{eq:scaling}) that $|\Omega| = 1$ and fill $\Omega$ with\footnote{Let us assume for the sake of simplicity that $ \beta $ is an integer number, otherwise one would have to take a number of balls equal to, e.g., the integer part of $ \beta $ and the computation would then become more involved, while the core of the proof would be unaffected.} $ N = \beta \gg 1$  disjoint balls of radius $1/\sqrt{\beta}$ centered at points $\rv_j \in \Omega$,
$$
	\B_j := \B_{1/\sqrt{\beta}}(\rv_j), 	\qquad j=1,\ldots,N.
$$
Let $f \in C_0^1(\B_1(0))$ be a radial function so that  $ \lf\| f \ri\|_2^2 = 1$, set
\beq
	u_j(\rv) := f \big(\sqrt{\beta}(\rv-\rv_j) \big) \in C_0^1(\B_j),
\eeq 
which clearly satisfies $ \lf\| u_j \ri\|_2^2 = 1/\beta$. Then the trial state we are going to use is
\beq
	\label{eq:trial state}
	u(\rv) :=  \sum_{j=1}^N u_j(\rv) e^{-i \sum_{k \neq j}\arg (\rv-\rv_k)} = \sum_{j=1}^N u_j(\rv) \prod_{k \neq j} \frac{z^* - z_k^*}{|z - z_k|},
\eeq
where for any point $ \rv = (x,y) $ in the plane we have used the complex notation $ z = x + i y $ and $ \arg z = \arctan \frac{y}{x} $. The properties of $ f $ and in particular its compact support contained in the unit ball imply that
$$
	|u(\rv)|^2 = \sum_{j=1}^N |u_j(\rv)|^2 =
	\begin{cases}
		|u_j(\rv)|^2,	& \text{on } \B_j, \\ 
		0,	&	\text{otherwise}.
	\end{cases}
$$
Then
\bmln{
	\cEAF_{\beta,\Omega} [u] 
	= \sum_{j=1}^N \int_{\B_j} \diff \rv \: \left| \left( -i\nabla + \beta\textstyle{\sum_{k=1}^N} \bA[|u_k|^2] \right) 
		e^{-i\sum_{k \neq j} \arg(\rv-\rv_k)} u_j \right|^2 \\
	= \sum_{j=1}^N \int_{\B_j} \diff \rv \: \left| \left( -i\nabla + \beta\bA[|u_j|^2] 
		+ \textstyle{\sum_{k \neq j}} \left( \beta\bA[|u_k|^2] - \nabla \arg(\rv-\rv_k) \right) 
		\right) u_j \right|^2 \\
	= \sum_{j=1}^N \int_{\B_j} \diff \rv \:  \left| (-i\nabla + \beta\bA[|u_j|^2]) u_j \right|^2,
}
where we used that by 2D Newton's theorem
\bmln{
	\bA[|u_k|^2](\rv) = \nablap \int_{\B_k} \diff \rvp \: \ln|\rv-\rvp| |u_k(\rvp)|^2 
	= \nablap \ln|\rv-\rv_k| \int_{\B_k} \diff \rvp \: |u_k(\rvp)|^2	\\
	= \frac{1}{\beta} \nabla \arg(\rv-\rv_k)
}
for $\rv \notin \B_k$.
Now note that for all $\rv \in \R^2$ and for all $ j = 1, \ldots, N $,
\bmln{
	\bA[|u_j|^2](\rv) = \int_{\B_j} \diff \rvp \: \frac{(\rv-\rvp)^\perp}{|\rv-\rvp|^2} |u_j(\rvp)|^2 
	 \\= \frac{1}{\sqrt{\beta}} \int_{\B_1} \diff \rvp \: \frac{(\sqrt{\beta}(\rv-\rv_j) - \rvp)^\perp}{|\sqrt{\beta}(\rv-\rv_j) -\rvp|^2} |f(\rvp)|^2 
	= \beta^{-\frac{1}{2}} \bA[|f|^2](\sqrt{\beta}(\rv-\rv_j))
}
and thus
\begin{align*}
	\cEAF_{\beta,\Omega}[u] 
	&= \sum_{j=1}^N \int_{\B_j} \diff \rv \:  \left| (-i\nabla u_j(\rv) + \beta\bA[|u_j|^2](\rv) u_j(\rv)) \right|^2  \\
	&= N \beta^{-1} \int_{\B_1} \diff \rvp \: \left| \left( -i \sqrt{\beta} \nabla f(\rvp) + \sqrt{\beta}\bA[|f|^2](\rvp) f(\rvp) \right) \right|^2  \\
	&= N \int_{\B_1} \diff \rv \: \left| (-i\nabla + \bA[|f|^2])f \right|^2 
	\ = N \cEAF_{1,\B_1}[f].
\end{align*}
Since $ \cEAF_{1,\B_1}[f] $ does not depend on $\beta$ we obtain the desired bound
$$
	E_0(\Omega,\beta,M) \leq C N = C \beta
$$
by the variational principle.	
	\end{proof}
	
	The bounds proven in Proposition \ref{pro:upper bound} and its proof are particularly interesting because of the form of the trial state \eqref{eq:trial state}: unlike the constant function heuristic mentioned in Section \ref{sec:homo} (see the discussion below Theorem \ref{teo:eaf homo} on $ e(1,1) $), the trial state has a very large phase circulation generated by singly-quantized vortices sitting at the centers of the balls $ \B_j $. As it is apparent from the energy computation, this huge circulation is in fact crucial in order to obtain an energy which depends linearly on $ \beta $. Note indeed that a naive approach would give an estimate of order $ \beta^2 $ for the energy $ \eaf_{\beta,\Omega} $ due to the square of the vector potential. This justifies once more the conjecture about the value of $ e(1,1) $.
	
	The density estimate \eqref{eq:uaf homo} is a straightforward consequence of the energy convergence. We skip most of the details and focus on the proof of Proposition \ref{pro:uaf homo 2}:
	
	\begin{proof}[Proof of Proposition \ref{pro:uaf homo 2}]
		Acting as in \cite[Proof of Lemma 4.1]{CLR} one can show that the energy estimate \eqref{eq:eaf homo} implies the following bound
		\beq
			\label{eq:l2 est}
			\lf\| \barho -  |\Omega|^{-1}  \ri\|_{L^2(\Omega)} = \OO(\beta^{-1/14+\eps}),
		\eeq
		where $ \barho $ is a suitable piecewise approximation of $ \uaf_{\beta, \mathrm{N}/\mathrm{D}} $, i.e., the result does not depend on boundary conditions. More precisely, $ \barho $ is constructed as follows: let $ \{ Q_j \}_{j \in J} $ a tiling covering $ \Omega $ made of squares $ Q_j $ of side length $ \beta^{-\nu} $, with
		\beq
			0 < \nu < \tx\frac{1}{2},
		\eeq
		then we set
		\beq
			\barho(\rv) : = \sum_{j \in \tilde{J}} \rho_j \one_{Q_j}(\rv),	\qquad		\rho_j : = \beta^{2\nu} \int_{Q_j} \diff \rv \: \big| \uaf_{\beta, \mathrm{N}/\mathrm{D}}(\rv) \big|^2,
		\eeq
		where for some $ 0 < \mu <1 - 2\nu $
		\beq
			\tilde{J} : = \lf\{ j \in J \: \big| \: \rho_j \geq \beta^{2\nu - 1 + \mu} \ri\},
		\eeq
		i.e., the mass of $ \uaf_{\beta, \mathrm{N}/\mathrm{D}} $ in the cells $ Q_j $, $ j \in \tilde{J} $, is not too small. 
		
		Let now $ \phi $ be a Lipschitz function with compact support contained in $ \B_R(0) $, then for any 
		\beq
			\eta < \nu,
		\eeq
		one has
		\bml{
			\int_{\B_R(0)} \diff \rv \: \phi(\rv)  \lf|\uaf_{\beta, \mathrm{N}/\mathrm{D}}\lf(\rv_0 + \beta^{-\eta} \rv\ri)\ri|^2 \\
			 = \beta^{2\eta} \sum_{Q_j \subset \B_{\beta^{-\eta} R}(\rv_0), \: j \in \tilde{J}} \int_{Q_j} \diff \rv \: \phi\lf(\beta^{\eta}(\rv - \rv_0)\ri) \lf|\uaf_{\beta, \mathrm{N}/\mathrm{D}}\lf(\rv\ri)\ri|^2 + \OO(\beta^{2(\eta - \mu)}) 	\\
			 = \beta^{2(\eta-\nu)} \sum_{Q_j \subset \B_{\beta^{-\eta} R}(\rv_0)} \phi\lf(\beta^{\eta}(\rv_j - \rv_0)\ri) \rho_j + \OO(\beta^{2(\eta - \nu)})		\\
			 = \int_{\B_R(0)} \diff \rv \: \phi(\rv)  \lf|\barho\lf(\rv_0 + \beta^{-\eta} \rv\ri)\ri|^2 + \OO(\beta^{2(\eta - \nu)}).	
		}
		Then
		\bml{
			\label{eq:cs}
			\lf| \int_{\B_R(0)} \diff \rv \: \phi(\rv) \lf[ \lf| \uaf_{\beta, \mathrm{N}/\mathrm{D}}\lf(\rv_0 + \beta^{-\eta} \rv\ri) \ri|^2 - |\Omega|^{-1} \ri] \ri| \\
			= \lf| \int_{\B_R(0)} \diff \rv \: \phi(\rv) \lf[ \barho\lf(\rv_0 + \beta^{-\eta} \rv\ri) -  |\Omega|^{-1} \ri] \ri| + o(1)	\\
			\leq C \beta^{\eta} \lf\| \phi \ri\|_{\mathrm{Lip}} \lf( \int_{\B_{\beta^{-\eta} R}(\rv_0)} \diff \rv \lf| \barho(\rv) - |\Omega|^{-1} \ri|^2 \ri)^{1/2}+ o(1)	\\
			\leq C \beta^{\eta} \lf\| \phi \ri\|_{\mathrm{Lip}} \lf\| \barho -  |\Omega|^{-1}  \ri\|_{L^2(\Omega)} \leq C \beta^{\eta - 1/14 + \eps} \lf\| \phi \ri\|_{\mathrm{Lip}}+ o(1),
		}
		and taking the supremum over Lipschitz functions with norm bounded by $ 1 $, one obtains the result, provided $ \eta < 1/14 $, which is allowed since $ \nu $ can be larger than $ 1/14 $.
	\end{proof}
	
\subsection{Trapped gas}
The proof of the local density approximation for the trapped gas relies heavily on the result for the homogeneous system discussed in the previous Section. The first step is as usual the energy estimate \eqref{eq:eaf}, which is proven by deriving suitable upper and lower bounds. In both cases the proof scheme is rather simple: by tiling the plane with squares of suitable length smaller than the TF scale, i.e., $ \beta^{1/(s+2)} $, but also much larger that the expected fine scale of the vortex lattice, i.e., $ \beta^{-s/(2(s+2))} $, one can use locally the result for the homogeneous gas. 

In the upper bound one has to get rid of the ``interaction'' between different cells, i.e., the magnetic field generated by the density in other cells, which is done by exploiting a trial state inspired by \eqref{eq:trial state}. The kinetic energy thus results in the sum of Dirichlet homogeneous energies inside the cells, up to remainder terms. The asymptotics \eqref{eq:eaf homo} together with a Riemann sum approximation of the potential term then yields the upper bound. 

The key step in the lower bound is on the other hand the use of the gauge invariance of the functional in order to cancel inside one given cell the magnetic potential generated by the density in all the other cells. The kinetic energy can then be bounded from below by the Neumann energy in a square. Riemann sum estimates and the scaling properties of the energy are then sufficient to complete the proof. Note that it is crucial to know that the Dirichlet and Neumann energies have the same limit as $ \beta \to \infty $.

The density estimate is proven by first deriving an $ L^2 $ estimate of the form \eqref{eq:l2 est} for a piecewise approximation of $ |\uaf_{\beta}|^2 $ and then using Cauchy-Schwarz as in \eqref{eq:cs}.

\bigskip

\noindent
{\footnotesize \textbf{Acknowledgments:} This work is supported by MIUR through the FIR grant 2013 ``Condensed Matter in Mathematical Physics (Cond-Math)'' (code RBFR13WAET), the Swedish Research Council (grant no. 2013-4734) and the ANR (Project Mathostaq ANR-13-JS01-0005-01). 
We thank Jan Philip Solovej for insightful suggestions and Romain Duboscq for inspiring numerical simulations. M.C. is also grateful to the organizing committee of the conference ``QMath13: Mathematical Results in Quantum Physics'', session ``Many-body Systems and Statistical Mechanics'', for the invitation to present this work there.}

\end{document}